\documentclass[aip,preprint]{revtex4-1}
\usepackage{amsmath, physics, amsthm, amsfonts, amssymb, latexsym, amscd, graphicx, enumitem}
\usepackage{xcolor,bm}
\draft 

\newcommand{\N}{\mathbb{N}}

\newcommand{\R}{\mathbb{R}}

\newtheorem{thm}{Theorem}[section]

\newtheorem{lemma}{Lemma}[section]
\newtheorem{prop}[thm]{Proposition}
\newtheorem{corollary}{Corollary}

 \theoremstyle{remark} \newtheorem{remark}{Remark}
\newtheorem{definition}{Definition}

\newtheorem{example}{Example}

\begin{document}


\title{Poisson-Dirac Submanifolds as a Paradigm for Imposing Constraints in Non-dissipative Plasma Models} 



\author{F.W.Pinto}
\email[]{fwp238@my.utexas.edu}
\affiliation{Department of Physics, The University of Texas at Austin, Austin, TX 78712}
\author{J.W.Burby}
\email[]{joshua.burby@austin.utexas.edu}
\affiliation{Department of Physics and Institute for Fusion Studies, The University of Texas at Austin, Austin, TX 78712}


\date{\today}

\begin{abstract}
    We present a generalisation of Dirac constraint theory based on the theory of Poisson-Dirac Submanifolds. The theory is formulated in a coordinate free manner while simultaneously relaxing the invertibility condition as seen in standard Dirac constraint theory. We illustrate the method with two examples: elimination of the electron number density using Gauss’ Law and Ideal MHD as a slow manifold constraint in the ideal two-fluid model.  
\end{abstract}


\maketitle 


\section{Introduction}

In order to derive reduced models without artificial dissipation, we often look to place constraints such that the reduced system inherits the ideal structure. The need for such techniques arise across physics (Gauge field theories, fluid dynamics, plasma physics) and dynamical systems theory (control theory \cite{Bloch2003}).  Robert Dewar explored the method of applying constraints to action principles  as a method for deriving dissipation free reduced models\cite{Dewar_2020, Dewar_Qu_2022}. In Ref.~\onlinecite{Dewar_Qu_2022}, Dewar applied constraints to bridge the gap between ideal Magnetohydrodynamics IMHD and Relaxed magnetohydrodynamics (RxMHD) by applying the ideal Ohm's law constraint to the RxMHD Lagrangian as a Lagrange multiplier to the action. One issue with this approach is that Lagrange multipliers can produce inconsistent Euler Lagrange equations. 
 
Dirac constraint theory (DCT) attacks the same problem from the Hamiltonian side, and looks for conditions on the constraints needed for the constrained system to be Hamiltonian. The method was originally introduced by Dirac\cite{Dirac_1950}, who developed the theory in an effort to quantize gauge field theories. More recently, the theory has been applied to the problem of imposing the incompressiblity constraint on the (compressible) Euler equations\cite{Nguyen_1999, Nguyen_2001}. The method was further developed by C. Chandre, P.J Morrison and others\cite{Chandre_2013, Morrison_2009, two_fluid_morri}.
 
The Dirac bracket used in the most recent work of Morrison\cite{Morrison_2020, 10.1063/1.3356103} provides an alternate method for imposing constraints that guarantee clean evolution equations. 

Additionally, constraint theory based on an alternative to Dirac brackets has been used to show how reduced theories like ideal magnetohydrodynamics (IMHD) inherit their Hamiltonian structure from the ideal two fluid Maxwell (ITFM)\cite{10.1063/1.4994068}.

The Dirac constraint method has two drawbacks. It first requires a coordinate system in which the constraint can be written as the vanishing of some set of coordinates $C=0$ and secondly requires that the Poisson bracket of the components $\{C_i,C_j\}$ is invertible \cite{Morrison_2020}. We describe a coordinate free geometrical method that relaxes the non degeneracy requirement by exploiting the theory of Poisson-Dirac Submanifolds. {  Marsden, Weinstein and others developed the local picture of Poisson geometry in the early to mid 1980's. In particular, Poisson-Dirac submanifolds appeared explicitly in Ref.~\onlinecite{Poisson_Red_Mars}. C. Chandre in Ref.~\onlinecite{CHANDRE20151} was able to relax the non-degeneracy requirement by using a psuedoinverse for $\{C_i,C_j\}$. Given the psuedoinverse is not unique infinite dimension, it is not clear that such a bracket is unique. Moreover, the proof of the Jacobi identity in Ref.~\onlinecite{CHANDRE20151} was coordinate dependent, and assumed that the constraints can be separated based on their class. The Poisson-Dirac construction resolves these issues and provides geometric avenue for constructing dissipation free brackets.}

In section 1, we outline the Poisson-Dirac Constraint Method as generalisation of DCT. We then make the the theory coordinate free by formulating the constraint sets as submanifolds of the phase space. In sections 2 and 3, we illustrate the method with two examples: elimination of the electron number density using Gauss' Law and Ideal MHD as a slow manifold constraint in the ideal two-fluid model. Both examples exhibit the constraint as the graph of a function although the theory is applicable more generally as long as the constrained system is a submanifold of the full system.

\section{Poisson-Dirac Constraint Method and Slow Manifolds}

We proceed by first recalling the assumptions of Dirac constraint theory. Let $M$ be the phase space with Poisson bracket $\{\cdot,\cdot \}$ and suppose that there are coordinates on $M$, say {  $(\sigma, c) = (\sigma_1,\dots \sigma_n, c_1,\cdots, c_m)$}, such that the constraint set of interest can be written in the form 
\begin{equation}\label{sigma} 
    \Sigma = \{(\sigma, c): c=0\}. 
\end{equation}
In order to derive a Poisson bracket on $\Sigma$, DCT assumes that the matrix of Poisson brackets
\begin{equation}
    \{c,c\}_{ij} = \{c_i,c_j\}
\end{equation}
is invertible. In such a case, there is a Poisson bracket \cite{Morrison_2020} on $\Sigma$ given by the formula
\begin{equation}
    \{f,g\}_{\text{D}} = \fdv{f}{\sigma}\cdot [\{\sigma,\sigma\}_\Sigma - \{\sigma,c\}_\Sigma\cdot \{c,c\}_\Sigma^{-1} \cdot \{c,\sigma\}_\Sigma]\cdot \fdv{g}{\sigma}.
\end{equation}
Here $f,g: \Sigma \to \R$, the $\Sigma$ subscript denotes restriction to $\Sigma$, and {  $\fdv{f}{\sigma}$ are vectors of derivatives of $f$ with respect to $\sigma$}. {  Moreover, the above bracket defines a valid Poisson bracket on the whole phase space, not just $\Sigma$}. This bracket, together with a Hamiltonian function on $\Sigma$, can be used to formulate Hamiltonian reduced models on $\Sigma$. The ``reduction" corresponds to passing from the ambient phase space $M$ to the constraint subspace $\Sigma$.

The invertibility requirement on $\{c,c\}$ is artificially restrictive. Suppose instead that the kernel of $\{c,c\}$ is a vector subspace of $\text{ker}\,\{\sigma,c\}$:
\begin{equation} \label{ker_condition}
    \ker \{c,c\} \subset \ker \{\sigma,c\}. 
\end{equation}
Under this new assumption, we can still define a bracket on $\Sigma$ using the formula
\begin{equation} \label{PD_bracket}
    \{f,g\}_{\text{PD}} = \fdv{f}{\sigma}\cdot \{\sigma,\sigma\}_\Sigma\cdot \fdv{g}{\sigma} - \frac{\Delta f}{\Delta c} \cdot \{c,c\}\cdot \frac{\Delta g}{\Delta c},
\end{equation}
 where the vector valued functions $\frac{\Delta f}{\Delta c}, \frac{\Delta g}{\Delta c}$ are (possibly non-unique) solutions of
\begin{align} \label{PD_eqns}
    \{c,\sigma\}\cdot \fdv{f}{\sigma} +  \{c,c\}\cdot \frac{\Delta f}{\Delta c} = 0\\
    \{c,\sigma\}\cdot \fdv{g}{\sigma} +  \{c,c\}\cdot \frac{\Delta g}{\Delta c} = 0.
\end{align}
The solvability conditions for these equations, namely
\begin{align*}
0&=v\cdot \{c,\sigma\}\cdot \fdv{f}{\sigma}\\
0&=v\cdot \{c,\sigma\}\cdot \fdv{g}{\sigma},
\end{align*}
for all $v\in\text{ker}\,\{c,c\}$, are automatically satisfied because $v\in \text{ker}\,\{c,c\}$ implies $v\in\text{ker}\,\{\sigma,c\}$ and $v\cdot \{c,\sigma\}\cdot w = - w\cdot \{\sigma,c\}\cdot v$. Solutions are determined modulo elements of $\text{ker}\{c,c\}$, but this non-uniqueness does not affect the value of the bracket $\{f,g\}_{\text{PD}}$ because $\frac{\Delta f}{\Delta c}, \frac{\Delta g}{\Delta c}$ are each contracted with $\{c,c\}$ in \eqref{PD_bracket}. {  We remark that \eqref{PD_eqns} are the same conditions for the existence for the existence of the psuedoinverse in Ref.~\onlinecite{CHANDRE20151}.} 

In fact, Eq.\,\eqref{PD_bracket} does indeed define a Poisson bracket on $\Sigma$, provided $\{f,g\}_{\text{PD}}$ is smooth whenever $f,g$ are smooth. To explain why, as well as to provide a coordinate-independent geometric formulation of \eqref{PD_bracket}, we will deduce \eqref{PD_bracket} from the theory of Poisson-Dirac submanifolds\cite{32155}. In so doing we will outline an extension of the Dirac constraint method that we term the Poisson-Dirac constraint method.

Before constructing the geometric picture, we give an example that demonstrates the Poisson-Dirac constraint method is strictly more general than the usual Dirac constraint method \cite{Morrison_2020}. 
\begin{example}
    Let the constraint coordinates $c$ be Casimirs of the Poisson bracket $\{\cdot,\cdot\}$ so that $\Sigma$ as defined in \eqref{sigma} is the level set of (some of the) Casimirs. Then, all the matrices $\{c,c\}, \{\sigma,c\}, \{c, \sigma\}$ vanish so that the kernel condition \eqref{ker_condition} is trivially satisfied. One notices that any $\frac{\Delta f}{\Delta c}, \frac{\Delta g}{\Delta c}$ are a valid solution to \eqref{PD_eqns}, Thus implying that the bracket on $\Sigma$ is
    \begin{equation}
        \{f,g\}_{Dirac} = \fdv{f}{\sigma}\cdot \{\sigma,\sigma\}_\Sigma\cdot \fdv{g}{\sigma}.
    \end{equation}
    Even though the inverse $\{c,c\}^{-1}$ does not exist. More generally, the constraint functions $c$ may functionally depend on Casimirs in a non-obvious way, leading to a non-zero $\{c,c\}$ that cannot be inverted.
\end{example}

We will now arrive at \eqref{PD_bracket} using a coordinate-independent argument. We first define the geometric analog of the Poisson bracket, the Poisson bivector. { Given a smooth manifold $M$, equipped with a Poisson bracket $\{\cdot,\cdot\}$, there is corresponding bivector $\pi$ defined as follows.}
\begin{definition}
    Given a Poisson bracket $\{\cdot,\cdot\}:C^\infty(M)\to C^\infty(M)$ on a smooth  manifold $M$, the \textbf{Poisson bivector} $\pi\in \mathcal{X}^2(M)$ is the unique bivector field such that 
    \begin{equation}
    \pi(df,dg) = \{f,g\}, 
    \end{equation}
    for all $f,g\in C^\infty(M)$.
\end{definition}
In a local coordinate patch on $M$, $\pi$ is written 
\begin{equation}
    \pi = \pi^{ij}\pdv{x^i}\wedge \pdv{x^j}.
\end{equation}
{  Here, we are assuming a sum over the indices $i,j\in \{1,\dots, n\}$. }The objects $\pdv{x^j}$ can be thought of as basis vectors dual to the basis co-vectors { $dx^j$}. We note one can formulate the condition for $\pi$ to satisfy a Jacobi identity by use of the \textbf{Schouten Bracket}\cite{32155, Marsden1999} of multi-vector fields. {  For a review on the calculus of multivector fields, see Ref.~\onlinecite{32155}.}

{ The bivector induces a map $\pi^\#:T^*M\rightarrow TM$ that sends each covector $\alpha$ to the vector $\pi^\#(\alpha)$ defined by requiring $\overline{\alpha}(\pi^\sharp(\alpha)) = \pi(\overline{\alpha},\alpha)$ for all covectors $\overline{\alpha}$ with the same basepoint as $\alpha$. In finite dimensions with coordinates $x^i$ the components of $\pi^\#$ are given by}\cite{32155, Marsden1999}:
\begin{equation}
    (\pi^\#)^{ij} = \{x^i,x^j\}. 
\end{equation}
One notices that if $f$ is any smooth function on $M$ then,
\begin{equation}
    \pi^\#(df) = X_f, 
\end{equation}
where $X_f$ denotes the Hamiltonian vector field with Hamiltonian $f$. For a comprehensive review on Cartan's calculus in mechanics and plasma physics, see Ref.~\onlinecite{Abraham1967Foundations} and Ref.~\onlinecite{MacKay_2020}. Naively, $df$ can be thought similarly to $(\nabla f)^T$ (in fact, they are equal up to an isomorphism that depends on one's choice of metric tensor). Geometers call the map $\pi^\#$ the \textbf{anchor map}\cite{32155, 10.1007/978-3-642-20300-8_20, fernandes2009momentummappoissongeometry}. Moreover, we can { recover the Poisson bivector from $\pi^\#$ by defining $\pi$ as the unique bivector that satisfies}
\begin{equation}
    \pi(df,dg) = \langle df,  \pi^\#(dg) \rangle.
\end{equation}
Here, $\langle\cdot,\cdot \rangle: T^*_pM\times T_pM\to \R$ denotes the non-degenerate pairing between the tangent space $T_pM$ and its dual $T^\ast_pM$\cite{RevModPhys.70.467}. 







Now we are ready to introduce the type of constraint that underlies Poisson-Dirac constraint theory. We consider constraints that define \emph{Poisson-Dirac submanifolds} \cite{32155}.
\begin{definition} \label{PD_def}
    A \textbf{Poisson-Dirac submanifold} $N\subset M$ of a Poisson manifold $(M,\pi)$ is a submanifold such that for all $p\in N$
    \begin{enumerate} [label=\textbf{PD.\arabic*},  ref=\arabic*]
        \item $T_pN\cap (T_pN)^{\pi\perp} = \{0\}$
        \item The induced bivector field on $N$ is smooth.
    \end{enumerate}
Here $(T_pN)^{\pi\perp}=\pi^\#(T_pN)^\circ $ is the $\pi$-orthogonal complement of $T_pN$, { and $(T_pN)^\circ :=\{\alpha_p\in T^\ast_pM: \alpha_p(\nu) = 0, \forall\nu\in T_pN
    \}\subset T_p^\ast M$ is the annihilator of $T_pN$}.
\end{definition}

\begin{remark}
The first condition implies that $\pi$ induces a bivector field on $N$ that is possibly non-smooth. Indeed, if $f,g$ are functions on $N$ and $p\in N$ let $\widetilde{(df_p)},\widetilde{(dg_p)}$ denote $1$-forms in $T^{\ast}_pM$ such that $\widetilde{(df_p)}\mid T_pN = df_p$, $\widetilde{(dg_p)}\mid T_pN = dg_p$, and $\widetilde{(df_p)}\mid (T_pN)^{\pi\perp} = \widetilde{(dg_p)}\mid (T_pN)^{\pi\perp} = 0$. {  Here, $\widetilde{(df_p)}\mid T_pN$ and $\widetilde{(df_p)}\mid (T_pN)^{\pi\perp}$ denote restriction of the covector $\widetilde{(df_p)}$ to the subspaces $T_pN$ and $(T_pN)^{\pi\perp}$ respectively}. The bracket between $f$ and $g$ at $p\in N$ is then defined as $\{f,g\}_N = \pi_p(\widetilde{(df_p)},\widetilde{(dg_p)})$. If $\widetilde{(df_p)}^\prime,\widetilde{(dg_p)}^\prime$ denote another pair of $1$-forms in $T_p^*M$ satisfying our conditions then $\widetilde{(df_p)}^\prime - \widetilde{(df_p)} \in (T_pN)^\circ$, which implies $\pi_p(\widetilde{(df_p)}^\prime,\widetilde{(dg_p)}^\prime) = \pi_p(\widetilde{(df_p)},\widetilde{(dg_p)})$.
\end{remark}
{  In order to gain a better intuition for \textbf{PD.1}, we provide a pictorial representation (given in Fig.1 below) and make analogy with Riemannian geometry. Let $(V,g)$ be a finite dimensional real vector space endowed with inner product $g:V\times V\to \R$. The orthogonal complement $W^\perp$ of a subspace $W\subset V$ is defined as 
\begin{equation}
    W^\perp = \{v \in V: g(v, w) = 0\quad \forall w\in W\}. 
\end{equation}
It is well known that 
\begin{equation}
    W\cap W^\perp = \{0\}.
\end{equation}
Moreover, there is a inner product defined on the dual space $V^\ast$ attained by lowering indices, which we denote by $h:V^\ast\times V^\ast\to \R$. It is straightforward to show $W^\perp= \tilde{h}(W^\circ)$, where $\tilde{h}:V^\ast\to V$ denotes the index raising map. Thus, $W\cap \tilde{h}(W^\circ) = \{0\}$, which bears a clear resemblance to the condition \textbf{PD.1}; replace $W$ with $T_pN$ and $\tilde{h}:V^*\rightarrow V$ with $\pi^\#:T^*_pM\rightarrow T_pM$.
A more detailed discussion of this point is given in Ref.~\onlinecite{32155}. When $\pi$ has no degeneracy, the condition reduces to the condition that $W$ be a symplectic subspace: $W\cap W^\omega = \{0\}$, where $W^\omega = \{v\in V: \omega(v,w)= 0\quad \forall w \in W\}$ is the $\omega$ orthogonal complement. 

At each point $p$ in a Poisson manifold the image of $\pi^\#$ is naturally a symplectic vector space. It is straightforward to show that \textbf{PD.1} implies the intersection of $T_pN$ with the image of $\pi^\#$ is a symplectic subspace.  We provide a low dimensional example of a Poisson-Dirac submanifold before proceeding.
\begin{figure}
    \centering
    \includegraphics[width=0.5\linewidth]{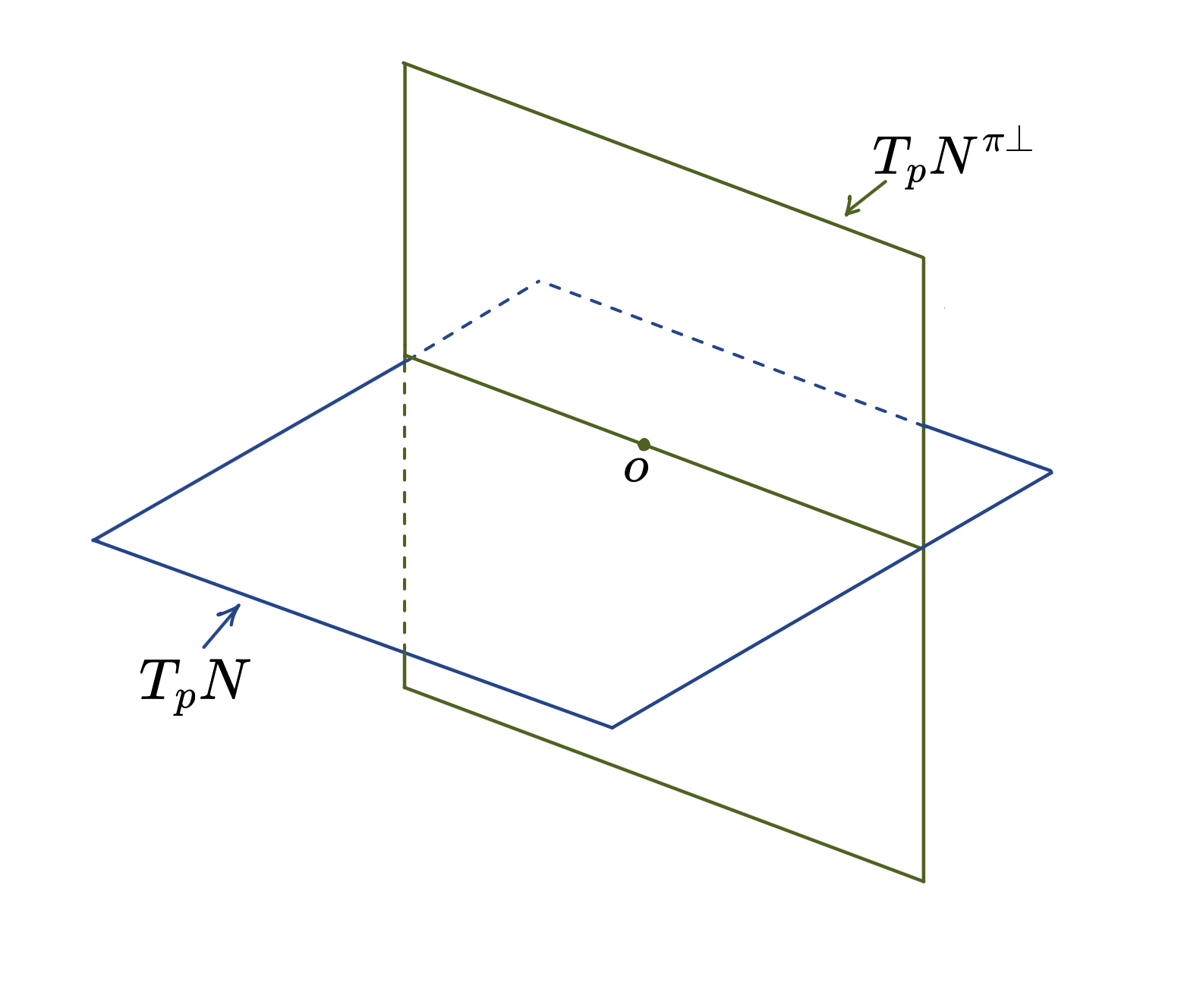}
    \caption{Visual representation of condition (1) in Def.\ref{PD_def}. The planes represent the vector spaces $T_pN$ and $(T_pN)^{\pi \perp}$ at a point $p \in  M$.}
    \label{fig:enter-label}
\end{figure}
\begin{example}
    Let $M= \R^4 \ni (x,y,z,w)$ with the Lotka-Volterra (LV) type Poisson structure\cite{32155}
    \begin{equation}
        \pi = xy\partial_x\wedge \partial_y + xw\partial_x\wedge \partial_w+zw\partial_z\wedge \partial_w + yz\partial_y\wedge \partial_z
    \end{equation}
    The submanifold given by the embedding 

    \begin{equation}
        (u,v)\mapsto (u,v,u,v)
    \end{equation}
    is a Poisson-Dirac submanifold, but not a Poisson transversal.
    \begin{proof}
        See Ref.~\onlinecite{32155} for details on the LV bracket. We proceed by characterising the spaces $T_pN$ and $(T_pN)^{\perp \pi}$. A simple calculation gives
        \begin{equation}
            T_pN = \{(\delta x,\delta y,\delta z,\delta w, p)\in \R^4\times \R^4: \delta z = \delta x,\delta w = \delta y\}
        \end{equation}
        The $\pi$ orthogonal is 
        \begin{equation}
            (T_pN)^{\perp\pi} = \{(\delta x,\delta y,\delta z,\delta w, p)\in \R^4 \times \R^4: \delta w = 0, \delta x = 0, \delta y = 2xy\alpha_x(p),\delta z = 2xy\alpha_y(p)  \}
        \end{equation}
        Where $p = (x,y,x,y)\in N$ is a point in $N$. Intersecting, we find that $\alpha_x(p),\alpha_y(p) = 0$ which implies that $T_pN\cap (T_pN)^{\perp \pi} = \{0\}$, making $N$ a Poisson-Dirac submanifold. Now, we compute $\pi_N$ by computing suitable extensions of forms $\alpha, \beta\in \Omega^1(N)$. Applying the condition that $\pi_N(\tilde{\alpha}), \pi_N(\tilde{\beta})\in T_pN$, $\tilde{\alpha}\mid N = \alpha$, and $\tilde{\beta}\mid N = \beta$, we find the induced bivector on $N$ is
        \begin{equation}
             \pi_N =  \frac{1}{2}uv\partial_u\wedge \partial_v. 
        \end{equation}
     \end{proof}
\end{example}
}
Using definition \ref{PD_def}, one can show that the following lemma holds \cite{32155}.
\begin{lemma}\label{Extension_condition}
    A submanifold $N\subset M$ is a Poisson-Dirac submanifold of a Poisson manifold $(M,\pi)$ if and only if for every one form $\alpha \in \Omega(N)$ there exists a smooth one form $\tilde{\alpha}\in \Omega(M)$ such that 
    \begin{enumerate}
        \item $\tilde{\alpha}|_{TN} = \alpha$
        \item ${\pi^\#}(\tilde{\alpha})$ is tangent to $N$.
    \end{enumerate}
    { Here, $\tilde{\alpha}|_{TN}$ means we are restricting $\tilde{\alpha}$ to points on $N$ and restricting the inputs to be vectors in $T_pN$.}
\end{lemma}
 Thus, one naturally obtains a Poisson structure $\pi_N$ on $N$ computed by evaluating extensions {  $\tilde{\alpha}, \tilde{\beta}\in \Omega^1(M)$ of one forms $\alpha, \beta\in \Omega^1(N)$ as prescribed above and restricting to $N$:}
\begin{equation} \label{PD_bracket_ext}
    \pi_N(\alpha, \beta) = \pi(\tilde{\alpha}, \tilde{\beta})|_N.
\end{equation}
{  For a coordinate free, dimension independent proof of the Jacobi identity for this bivector, see Ref.~\onlinecite{32155}. 
\begin{figure}
    \centering
    \includegraphics[width=0.5\linewidth]{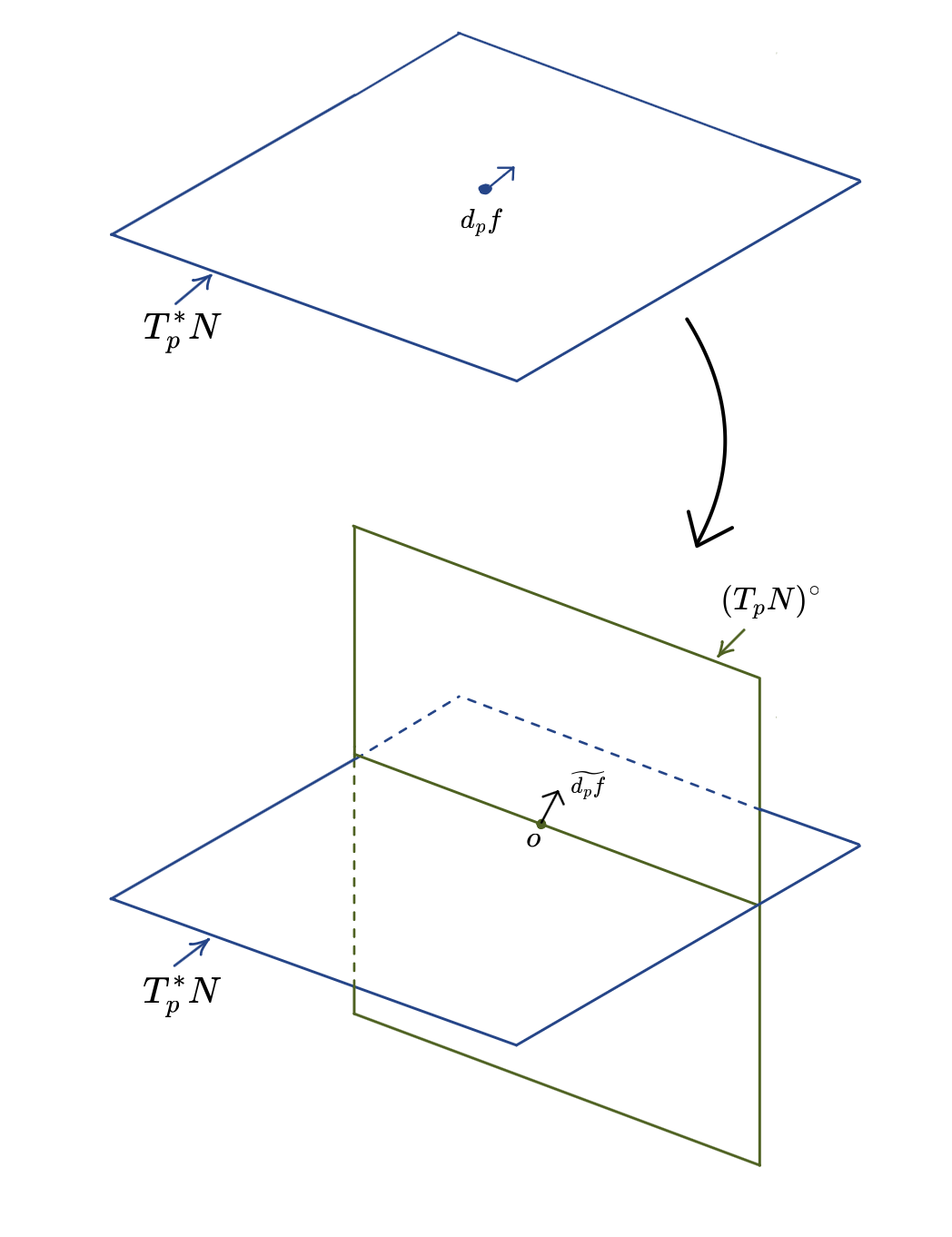}
    \caption{Visual representation of Lemma 2. The planes represent the vector spaces $(T_pN)^\circ$ and $T^\ast_pN$ at a point $p \in  M$.}
    \label{fig:extend_Func}
\end{figure}
\begin{remark}
    Notice that such extensions are guaranteed by the first condition in definition \ref{PD_def} in the definition of a Poisson-Dirac submanifold. 
    Thus, the condition allows us to find extensions in $T^\ast M$ that vanish on $(T_pN)^{\pi\perp}$. Notice that condition implies that $T_pM =T_pN \oplus (T_pN)^{\pi\perp} \oplus R_p$ where $R_p$ is some residual subspace that contains vectors not in $T^\ast_pN$ or $(T_pN)^{\pi\perp}$. There is no obvious block diagonal structure for $\pi^\#$ when expressed with respect to direct sum decomposition $T_pN \oplus (T_pN)^{\pi\perp} \oplus R_p$, in contrast to what happens with Poisson transversals (i.e. when $R_p = \{0\}$). A visual representation of the extension procedure is given in figure \ref{fig:extend_Func} above. 
\end{remark}
} 
To relate the general theory of Poisson-Dirac submanifolds to the coordinate-based DCT discussed earlier, we will now consider phase spaces of the form $\Sigma\times C$ where $\Sigma,C$ are vector spaces. The Poisson anchor map on such a space takes the convenient block form:
\begin{equation}
    \pi^\#= \begin{pmatrix}
        \pi^\#_{cc} & {\pi}^\#_{c\sigma}\\
        {\pi}^\#_{\sigma c} & {\pi}^\#_{\sigma \sigma}
    \end{pmatrix},
\end{equation}
where each ${\pi^\#}_{ij}$ denotes a linear map where the subscripts  $j$, $i$ (being either $C$ or $\Sigma$) denote the source and target respectively, and { 
\begin{equation}
    ((\pi^\#)^{ij})^\ast = -(\pi^\#)^{ji}
\end{equation}
Where the upper star denotes the dual map. In general, the diagonal terms in the above Poisson bracket are not necessarily zero.}
With all the relevant definitions stated, we demonstrate that \eqref{ker_condition} holds if and only if $\Sigma$ is a Poisson-Dirac submanifold. 
\begin{prop}
     $\Sigma$ is a Poisson-Dirac submanifold if and only if $\ker \{c,c\} \subset \ker \{\sigma,c\}$ and the bracket defined in equation \eqref{PD_bracket} is smooth.
\end{prop}
\begin{proof}
    Let us formulate the Poisson-Dirac condition algebraically. Suppose that $\nu\in T_p\Sigma\cap (T_p\Sigma)^{\perp \pi}$, and let $\pi^\#_{cc}:C^\ast \to C$, $\pi^\#_{\sigma c}:C^\ast \to T\Sigma$ be the components of the Poisson anchor with source $C^\ast$. Given $\nu \in (T_p\Sigma)^{\perp \pi}$, there exists an $\alpha_c \in T_p^\ast C$ with 
    \begin{equation}
        \nu =\pi^\#_{cc} (\alpha_c) + \pi^\#_{\sigma c} (\alpha_c)
    \end{equation}
    enforcing the condition $\nu\in T_p\Sigma$ we require that $\alpha_c\in \ker \pi^\#_{cc}$. In order for $\Sigma$ to be Poisson Dirac, we require that $\alpha_c\in \ker\pi^\#_{\sigma c}$. Suppose towards a contradiction that $\ker \pi^\#_{c c}\not\subset \ker \pi^\#_{\sigma c}$. Then there exist elements of $\ker \pi^\#_{cc}$ that are not contained within $\ker \pi^\#_{\sigma c}$, which would contradict that the only the zero vector is contained within the intersection.

    Going in the other direction, Given $\nu \in (T_p\Sigma)^{\perp \pi}$, there exists a $\alpha_c \in T_p^\ast C$ with 
    \begin{equation}
        \nu =\pi^\#_{cc} (\alpha_c) + \pi^\#_{\sigma c} (\alpha_c)
    \end{equation}
    Suppose that $\nu$ is non zero and lies within the intersection $(T_p\Sigma)^{\perp \pi}\cap T_p\Sigma$. Given $\nu\in T_p\Sigma$, this forces $\pi^\#_{cc} (\alpha_c) = 0$ which implies $\alpha_c\in \ker \pi^\#_{cc}$. But $\ker \pi^\#_{cc} \subset \ker \pi^\#_{\sigma c}$, so $\nu = 0$. Given $\nu$ was arbitrary, we conclude that $(T_p\Sigma)^{\perp \pi}\cap T_p\Sigma= \{0\}$. Finally, the bracket defined in equation \eqref{PD_bracket} is the induced bracket on $\Sigma$ so condition \textbf{PD.1} holds if and only if \eqref{PD_bracket} is smooth.

    Thus, the necessary and sufficient condition that $\Sigma$ is Poisson-Dirac is given by the kernel nesting condition along with a smoothness condition on the induced pointwise bivector. 
\end{proof}
{  
\begin{prop}
    The bracket \eqref{PD_bracket} is equivalent to the one obtained in Lemma \ref{Extension_condition}.
\end{prop}
\begin{proof}
    Now, let $f,g\in C^{\infty}(\Sigma)$ be smooth functions on $\Sigma$. Admissible extensions are of the general form $\tilde{f}(c,\sigma) = f(\sigma) + h(c), \tilde{g}(c,\sigma) = g(\sigma) + r(c)$. Taking the exterior derivative and applying $\pi^\#$, we find that the equations to solve for $h(c)$ are given by
    \begin{equation}
        \pi^\#_{c\sigma} (df) + \pi^\#_{cc}(dh) = 0
    \end{equation}
    which are exactly the equations given in \eqref{PD_eqns}. Given the submanifold is Poisson-Dirac, the solvability conditions are satisfied, and we write the Poisson bracket on $\Sigma$ 
    \begin{equation}
        \pi_{\Sigma}(df,dg) = \pi(d\tilde{f}, d\tilde{g})\mid_{\Sigma} = \langle df, \pi^\#_{\sigma\sigma}(dg)\rangle\mid_{\Sigma} - \langle dh, \pi_{cc}^\#(dr)\rangle.
    \end{equation}
    As expected, this matches the bracket given in equation \eqref{PD_bracket}.
\end{proof}
}

Translating DCT into the language of geometry, we see that DCT is a special type of Poisson-Dirac submanifold called a \emph{Poisson transversal}\cite{32155, AW_LSoPM}.
\begin{definition}
    A \textbf{Poisson transversal} $N\subset M$ of a Poisson manifold $(M,\pi)$ is a submanifold such that for all $p\in N$
    \begin{enumerate} [label=\textbf{PT.\arabic*}]
        \item $T_pN\cap (T_pN)^{\pi\perp} = \{0\}$
        \item $T_pN + \Im\pi_p^\# = T_pM$.
    \end{enumerate}
\end{definition}
{   
\begin{remark}
    Note that the submanifold as given in the definition above $N$ cannot be a Poisson submanifold, as the transversality condition is not satisfied for Poisson submanifolds. Even though Poisson submanifolds are not Poisson transversals, they are Poisson-Dirac submanifolds as $(T_pN)^{\pi \perp }=\{0\}$. 
\end{remark}
\begin{remark}
    It is worth noting that when $\pi$ has non constant rank, the foliation by symplectic leaves is potentially singular. The above definition does not imply that the dimension of the image of $\pi$ is constant.
\end{remark}
}
Transversality enforces several important properties such as the smoothness of $\pi$ restricted to $N$. In the below proposition, we demonstrate that the invertibility condition on $\pi_{cc}^\#$ is satisfied if and only if $\Sigma$ is a Poisson transversal.
\begin{prop}
    Within standard DCT, $\{c,c\}$ is invertible if and only if the constraint set is a Poisson transversal.
\end{prop}
\begin{proof}
    In the coordinates $(\sigma,c)$, the tangent space to the constraint is given $T_p\Sigma = \{(\delta \sigma, \delta c): \delta c=0\}$. Assuming that the matrix $\{c,c\}_{ij}$ is invertible, then the $\pi^\#_{cc}:C^\ast \to C$ section of the block matrix must be invertible, so that $T_p\Sigma + T_pL_p = T_pM$. Suppose $\nu\in T_p\Sigma\cap T_p\Sigma^{\perp \pi}$. Then, there exists an $\alpha \in T_p\Sigma^{\circ} = T_p^\ast C$ such that 
    \begin{equation}
        \nu = \pi^\#_{cc}(\alpha) + \pi^\#_{\sigma c}(\alpha).
    \end{equation}
    Given $\nu\in T_p\Sigma$,  $\pi^\#_{cc}\alpha$ must be zero. Given $\pi^\#_{cc}$ is invertible, the only solution is $\alpha = 0$ so that $\nu =0$ as required. Going the other direction, the transversality condition implies that $\pi^\#_{cc}$ must be surjective. Given $\Im  (\pi^\#_{cc})^\ast \approxeq (\ker\pi^\#_{cc})^\circ$, and $\pi^\#_{cc}$ is {  skew-symmetric}, $\ker \pi^\#_{cc} = \{0\}$, making $\pi^\#_{cc}$ an isomorphism.  
\end{proof}
{  From this geometric standpoint, it becomes more clear why the bracket defined in equation (3) is valid on the whole phase space, not just $\Sigma$. When the constraint manifold is given by a regular level set of a function, transversality implies level sets neighboring the constraint set are also Poisson transversals with their own induced Poisson tensors. The collection of symplectic leaves on the level sets neighboring the constraint define a symplectic foliation, and therefore a Poisson tensor extending the constraint Poisson tensor, in an open neighborhood around $\Sigma$. Notice that this argument also immediately implies that the constraint functions are Casimirs for the bracket defined in the neighborhood of the constraint. 

When $\Sigma$ is Poisson-Dirac, the rank of $T_p\Sigma^{\perp \pi}$ is possibly non constant, and this argument for extending the Poisson tensor off of the constraint manifold can fail. We return to the LV bracket example to better illustrate this point. Additionally, we show that it is not always possible to extend the Poisson bivector on $\Sigma$ in an open neighborhood around $\Sigma$ when $\Sigma$ is Poisson-Dirac.
\begin{example}
    Going back to the LV bracket in the previous example, let us try to construct a bivector in a neighborhood $U\in\R^4$ of the origin. Recall that the bivector on $N$ is given by 
    \begin{equation}
        \pi_N = \frac{1}{2}uv \partial_u\wedge \partial_v
    \end{equation}
    The submanifold $N$ may be written as the level set of the constraint functions
    \begin{equation}
        N=\{x\in \R^4: f_1(x) = 0, f_2(x) = 0\}
    \end{equation}
    where
    \begin{align}
        f_1(x) = x-z \\
        f_2(x) = y-w
    \end{align}
    A direct computation verifies that these constraints are second class constraints when $x \neq 0$, but become on Casimirs at the origin. As a result, the matrix of the brackets second class constraints becomes singular at the origin. The pseudo inverse $\mathbb{D}$ of the matrix $\mathbb{C}$ proposed in Ref.~\onlinecite{CHANDRE20151} fails to be smooth as the pseudoinverse $\mathbb{D}$ at the origin discontinuously drops to zero. Moreover, the rank of $\pi_N$ drops to from $2$ to $0$ at the origin, so that any candidate extension, $(\tilde{\pi},U)$ of $\tilde{\pi}_N$ will fail to be smooth. See the appendix section and Ref.~\onlinecite{AW_LSoPM} for more information on this. In essence, the smoothness of $\pi$ does not prevent rank dropping events, but it does prevent rank jumping events. In essence, the proof of the Jacobi identity given in Ref.~\onlinecite{CHANDRE20151} relied on a change of coordinates that tacitly assumes that $\pi$ has constant rank which in general does not hold.
\end{example}
}
For our subsequent applications of Poisson-Dirac constraint theory, we now summarize the key elements of the formalism when applied to submanifolds $N$ given as graphs of functions. Suppose phase space is given by the product $X\times Y\ni (x,y)$. (We distinguish notationally between $X,Y$ and $\Sigma,C$ in order to remind readers that neither of the variables $(x,y)$ is required to be comprised of constraints, as is true for $c$.) Let $N$ be the graph of a function $y_0^*:X\rightarrow Y$. 
The tangent space to $N$ at $p\in N$ is
\begin{equation}
    T_pN = \{(p, (\delta x,\delta y): \delta y = Dy_0^*(p)[\delta x]\}.
\end{equation}
The the annihilator space at $p$ is
\begin{equation}
\begin{split}
    (T_pN)^\circ = \{\alpha= (\alpha^X,\alpha^Y)\in X^\ast \times Y^\ast:\\ \alpha^X(\delta x)= -\alpha^Y(Dy^\ast_0(p)[\delta x]) \}.
\end{split}
\end{equation}
The space ${\pi^\#}(T_pN)^\circ$ is given by
\begin{equation}
\begin{split}
    {\pi}^\#(T_pN)^\circ = \{(\delta x, \delta y):\exists \alpha^Y\in Y^\ast \, \, \text{such that}\\ \delta x = (-{\pi}^\#_{xx}(Dy_0^\ast(x))^\ast + {\pi}^\#_{xy})\alpha^Y , \\ \delta y = (-{\pi}^\#_{yx}(Dy_0^\ast(x))^\ast + {\pi}^\#_{yy})\alpha^Y \}.
\end{split}
\end{equation}
Finally, given that $N = \Gamma(y_0^\ast)$ is the graph of a function, we will assume the ansatz for the extension 
\begin{equation} \label{Extension_ansatz}
    \tilde{\alpha} = \alpha + \langle \Theta(x),dy - Dy_0^\ast[dx] \rangle. 
\end{equation}
One notices that $\tilde{\alpha}\mid_{N} = \alpha$ thus making $\tilde{\alpha}$ automatically satisfy one of the conditions needed to be a suitable extension as in lemma \eqref{Extension_condition}. With all relevant background covered, we proceed with the Gauss law constraint. {  Given that we have formulated the theory in a coordinate free, dimension independent manner, there are no \emph{formal} obstructions to employing the theory in infinite dimensions. (Here ``dimension independent" refers to an argument that only uses results that apply to both finite-dimensional and infinite-dimensional spaces.) Importantly, the proof that the bivector defined in equation \eqref{PD_bracket_ext} satisfies the Jacobi identity is coordinate free and dimension independent. For more information on infinite dimensional Poisson geometry, see Ref.~\onlinecite{MR836734, RevModPhys.70.467, marsden1983hamiltonian}. In short, three key changes are made when moving to infinite dimension; we replace functions with functionals, exterior derivatives with functional derivatives, and the operator $\pi^\#$ is more generally a linear operator.}

\section{Ideal Two Fluid Maxwell With Gauss Law Constraint}

We will now consider two applications of Poisson-Dirac constraint theory in the context of the two-fluid Maxwell system. The equations for the asymptotically scaled ideal two fluid Maxwell system with domain being the three torus $\mathbb{T}^3$ are \cite{10.1063/1.4994068}

\begin{gather} 
\hspace{-4em} m_i n_i (\partial_t\boldsymbol{u}_i+\boldsymbol{u}_i\cdot\nabla\boldsymbol{u}_i)=-\nabla \mathsf{p}_i\label{ion_momentum}\\
\hspace{8em}-\frac{1}{\epsilon}Z_iq_e n_i (\boldsymbol{E}+\boldsymbol{u}_i\times\boldsymbol{B})\nonumber\\
\hspace{-4em}m_en_e(\partial_t\boldsymbol{u}_e+\boldsymbol{u}_e\cdot\nabla\boldsymbol{u}_e)=\label{electron_momentum}\\
\hspace{1em}-\nabla\mathsf{p}_e+ \frac{1}{\epsilon}q_e n_e(\boldsymbol{E}+\boldsymbol{u}_e\times\boldsymbol{B})\nonumber\\
\partial_t n_i+\nabla\cdot(n_i \boldsymbol{u}_i)=0\label{ion_continuity}\\
\partial_t n_e+\nabla\cdot(n_e \boldsymbol{u}_e)=0\label{electron_continuity}\\
\nabla\times\boldsymbol{B}=-\frac{1}{\epsilon}\mu_oq_eZ_i n_i\boldsymbol{u}_i + \frac{1}{\epsilon}\mu_oq_e n_e\boldsymbol{u}_e\label{ampere}\\
\hspace{8em}+\epsilon\,\mu_o\epsilon_o\partial_t\boldsymbol{E}\nonumber\\
\nabla\times\boldsymbol{E}=-\partial_t\boldsymbol{B}.\label{faraday}
\end{gather}
Here $n_i$ is the ion number density, $\boldsymbol{u}_\sigma$ is the velocity of species $\sigma$, $\mathsf{p}_\sigma$ is the species-$\sigma$ partial pressure, $\boldsymbol{B}$ is the divergence-free magnetic field, $\boldsymbol{E}$ is the electric field, $m_\sigma$ is the species-$\sigma$ mass, $Z_i$ is the ionic atomic number, $q_e$ is (minus) the elementary unit of charge, and $\mu_0,\epsilon_0$ are the MKS vacuum permeability and permittivity\cite{10.1063/1.4994068}. To proceed, we present the Poisson bracket for this system using the phase space Lagrangian given in Ref.~\onlinecite{10.1063/1.4994068} (see Ref.~\onlinecite{osti_1878066} for more details systematically deriving Poisson brackets from Lagrangians). 

\begin{lemma}
    The two fluid Maxwell dynamics  \eqref{ion_momentum}-\eqref{faraday} comprise a Poisson system with Hamiltonian given in Ref.~\onlinecite{10.1063/1.4994068}.
\end{lemma}

\begin{proof}
    This result is well known (see Ref.~\onlinecite{PhysRevA.25.2437,marsden1983hamiltonian}, for alternate derivations). By considering a variation of the above action in Ref.~\onlinecite{10.1063/1.4994068}, we arrive at the two fluid Poisson bracket 
    \begin{equation}
\begin{split}
    \{F,G\}(\boldsymbol{E},\boldsymbol{B},n_i,n_e,\boldsymbol{u}_e,\boldsymbol{u}_i)= \\\sum_{\sigma \in \{i,e\}}\int  -m_\sigma^{-1}(F_{n_\sigma} \nabla \cdot G_{\boldsymbol{u}_{\sigma}} + F_{\boldsymbol{u}_{\sigma}} \cdot \nabla G_{n_\sigma})\,d^3\boldsymbol{x} \\+ \int \frac{\nabla \times \boldsymbol{u}_{\sigma}}{m_\sigma n_\sigma }\cdot (F_{\boldsymbol{u}_{\sigma}}\times G_{\boldsymbol{u}_{\sigma}})\,d^3\boldsymbol{x} \\ \int \frac{1}{\epsilon \epsilon_o}(F_{\boldsymbol{E}} \cdot (\nabla\times{G_{\boldsymbol{B}}}) - G_{\boldsymbol{E}} \cdot (\nabla\times{F_{\boldsymbol{B}}}))\,d^3\boldsymbol{x}  \\+ \int  \frac{q_\sigma}{m_\sigma \epsilon^2 \epsilon_o}(F_{\boldsymbol{u}_{\sigma}}\cdot G_{\boldsymbol{E}} - F_{\boldsymbol{E}}\cdot G_{\boldsymbol{u}_{\sigma}})\, d^3\boldsymbol{x} \\ + \int \frac{q_\sigma \boldsymbol{B}}{\epsilon(m_\sigma)^2 n_\sigma}  \cdot (F_{\boldsymbol{u}_{\sigma}}\cross G_{\boldsymbol{u}_{\sigma}})\,d^3\boldsymbol{x}
\end{split}
\end{equation}
where $q_i = -Z_iq_e$ and $q_e$ are the charges of the ions and electrons respectively. Additionally, we use the compact functional derivative notation of Morrison\cite{RevModPhys.70.467} $F_{\chi_i} = \fdv{F}{\chi_i}$, where $\chi_i\in (\boldsymbol{E},\boldsymbol{B},n_i,n_e,\boldsymbol{u}_e,\boldsymbol{u}_i) := M$ is some field variable. This defines the Poisson bracket, making the two-fluid Maxwell dynamics Poisson as claimed. {For a proof of the Jacobi identity, see Ref.~\onlinecite{PhysRevA.25.2437}.}
\end{proof}
This bracket is defined on a space that includes $n_e$ as an independent variable. We can obtain a an alternative, smaller formulation that does not involve $n_e$ by restricting to the submanifold defined by the Gauss constraint $\epsilon\,\epsilon_0\,\nabla\cdot \bm{E} = q_e\,n_e + q_i\,n_i$. We perform this reduction in the next Section using Poisson-Dirac constraint theory.

\section{Gauss law Casimir is a Poisson-Dirac Submanifold}

\begin{lemma}
    The subset $N$ of $(\bm{E},\bm{B},n_i,n_e,\bm{u}_e,\bm{u}_i)$-space defined by the Gauss constraint is a Poisson-Dirac submanifold.
\end{lemma}

\begin{proof}
    The Poisson bracket in equation (251) induces a bundle map $\pi^{\#}:M^\ast \to M$ given by
    \begin{widetext}
    \begin{equation}
        \pi^{\#}(\delta F)=X_{\delta F} = \begin{pmatrix}
            -\frac{1}{\epsilon \epsilon_o}\nabla\times{F_{\boldsymbol{E}}} \\
            \frac{1}{\epsilon \epsilon_o}\nabla\times{F_{\boldsymbol{B}}} - \frac{q_e}{\epsilon^2 \epsilon_o m_e}F_{\boldsymbol{u}_e} +  \frac{Z_iq_e}{\epsilon^2 \epsilon_o m_i}F_{\boldsymbol{u}_i} \\
            -\frac{Z_iq_e}{\epsilon m_i^2n_i}(F_{\boldsymbol{u}_i}\times B) -\frac{Z_iq_e}{\epsilon^2 \epsilon_o m_i}F_{\boldsymbol{E}}-\frac{1}{m_i}\nabla F_{n_i} + \frac{1}{m_i n_i} F_{\boldsymbol{u}_{i}} \times (\nabla \times \boldsymbol{u}_{i}) \\ 
            \frac{q_e}{\epsilon m_e^2n_e}(F_{\boldsymbol{u}_e}\times B) +\frac{q_e}{\epsilon^2 \epsilon_o m_e}F_{\boldsymbol{E}}-\frac{1}{m_e}\nabla F_{n_e} + \frac{1}{m_e n_e} F_{\boldsymbol{u}_{e}} \times (\nabla \times \boldsymbol{u}_{e}) \\
            - m_i^{-1}\nabla\cdot{F_{\boldsymbol{u}_i}} \\
            - m_e^{-1}\nabla\cdot{F_{\boldsymbol{u}_e}}
        \end{pmatrix}= \begin{pmatrix}
            \delta {\boldsymbol{B}}\\
            \delta {\boldsymbol{E}}\\
            \delta {\boldsymbol{u}_i}\\
            \delta {\boldsymbol{u}_e}\\
            \delta {n_i}\\
            \delta {n_e}
        \end{pmatrix}
    \end{equation}
    \end{widetext}
    where $\delta F$ is an arbitrary linear functional.
    Notice that $T_pN = \{(\delta \boldsymbol{B}, \delta \boldsymbol{E}, \delta \boldsymbol{u}_{i}, \delta \boldsymbol{u}_{e}, \delta n_i, \delta n_e): \delta n_e = Z_i \delta n_i + \epsilon^2 \epsilon_0q_e^{-1}\nabla\cdot\delta\boldsymbol{E}\}$ spans all elements of the two fluid space other than the electron number density. However, not all scalar functions are the divergence of vector valued functions, so that $N$ is not transverse to the leaf making the DCT method inapplicable. In order to prove $T_pN\cap (T_pN)^{\pi\perp}= \{0\}$, suppose that there exists a non-zero vector $\delta F\in T_pN\cap (T_pN)^{\pi\perp}$ contained within the intersection. We characterise the $\pi$-orthogonal of $T_pN$ using formulae from section 1:
    \begin{equation}
    \begin{split}
        \delta F &= \langle F_{n_e}, \delta n_e\rangle - \langle F_{n_e},Z_i \delta n_i + \epsilon ^2 \epsilon_oq_e^{-1} \nabla \cdot \delta \boldsymbol{E}\rangle   \\&=\langle F_{n_e}, \delta n_e\rangle - \langle F_{n_e},Z_i \delta n_i \rangle + \langle \nabla F_{n_e},\epsilon ^2 \epsilon_oq_e^{-1} \delta \boldsymbol{E}\rangle
    \end{split}
    \end{equation}
    where $\langle\cdot ,\cdot \rangle: M\times M^\ast\to \R$ is the natural pairing as defined in the previous section. Applying this to the bundle map, we get the zero vector:
    \begin{align}
        \begin{pmatrix}
            0\\
            0\\
            \frac{-Z_iq_e}{\epsilon^2 \epsilon_o m_i}(\epsilon ^2 \epsilon_oq_e^{-1} \nabla F_{n_e})+m_i^{-1}Z_i\nabla F_{n_e}\\
            +\frac{q_e}{\epsilon^2\epsilon_o  m_e}(\epsilon ^2 \epsilon_oq_e^{-1}\nabla F_{n_e}) - m_e^{-1}\nabla F_{n_e}\\
            0\\
            0
        \end{pmatrix}=\begin{pmatrix}
            0\\
            0\\
            0\\
            0\\
            0\\
            0
        \end{pmatrix},
    \end{align}
    which implies that $(T_pN)^{\pi\perp} = \{0\}$, making $N$ a Poisson-Dirac submanifold.
    \end{proof}
    
Now that $N$ is a Poisson Dirac submanifold, we find the bracket on $N$. Requiring that $\pi^\#(\tilde{F})$ be tangent to $N$ is trivial since the $\pi$-orthogonal contains only the zero vector. The relevant extensions may be found in slightly different ways. Using eqn.\eqref{Extension_ansatz} in the previous section, we write the ansatz for an extension, $\delta \tilde{F}\in TM^\ast$ of a one form $\delta F\in T^\ast N$
\begin{equation}
\begin{split}
    \delta\tilde{F} = \int (F_{n_i} - Z_i\Theta_{n_e})\delta n_id^3\boldsymbol{x} + \int F_{\boldsymbol{B}} \cdot \delta \boldsymbol{B} \,d^3\boldsymbol{x} \\+ \int(F_{\boldsymbol{E}}+\epsilon^2\epsilon_o q_e^{-1}\nabla\Theta_{n_e})\cdot \delta \boldsymbol{E}\, d^3\boldsymbol{x} + \int\Theta_{n_e}\cdot \delta{n_e}\, d^3\boldsymbol{x}\\ + \int F_{\boldsymbol{u}_i}\cdot \delta{\boldsymbol{u}_i}\,d^3\boldsymbol{x} + \int F_{\boldsymbol{u}_e}\cdot \delta{\boldsymbol{u}_e}\,d^3\boldsymbol{x}.
\end{split}
\end{equation}
Now, we apply the above extension to $\pi$:
Notice that all $\Theta$ terms vanish from the above vector field, so that any choice for $\Theta_{n_e}$ is valid. Thus, we may set $\Theta = 0$ resulting in the Poisson tensor $\pi_{N}$ on $N$:
\begin{widetext}
    \begin{equation}
    \pi^\#_{N}(\delta F)=X_{\delta F} = \begin{pmatrix}
            -\frac{1}{\epsilon \epsilon_o}\nabla\times{F_{\boldsymbol{E}}} \\
            \frac{1}{\epsilon \epsilon_o}\nabla\times{F_{\boldsymbol{B}}} - \frac{q_e}{\epsilon^2 \epsilon_o m_e}F_{\boldsymbol{u}_e} +  \frac{Z_iq_e}{\epsilon^2 \epsilon_o m_i}F_{\boldsymbol{u}_i} \\
            -\frac{Z_iq_e}{\epsilon m_i^2n_i}(F_{\boldsymbol{u}_i}\times B) -\frac{Z_iq_e}{\epsilon^2 \epsilon_o m_i}F_{\boldsymbol{E}}-m_i^{-1}\nabla F_{n_i} + \frac{1}{m_i n_i} F_{\boldsymbol{u}_{i}} \times (\nabla \times \boldsymbol{u}_{i}) \\ 
            \frac{q_e}{\epsilon m_e^2n_e}(F_{\boldsymbol{u}_e}\times B) +\frac{q_e}{\epsilon^2 \epsilon_o m_e}F_{\boldsymbol{E}} + \frac{1}{m_e n_e} F_{\boldsymbol{u}_{e}} \times (\nabla \times \boldsymbol{u}_{e}) \\
            - m_i^{-1}\nabla\cdot{F_{\boldsymbol{u}_i}} \\
        \end{pmatrix}
\end{equation}
\end{widetext}
By writing $\langle \delta F, \pi^\#_{N}(\delta G) \rangle$, we obtain the Poisson bracket on $N$: 
\begin{align}
    \begin{split}
    \{F,G\}_{N}  = \sum_{\sigma \in \{i,e\}}\int  -m_i^{-1}(F_{n_i} \nabla \cdot G_{\boldsymbol{u}_{i}} + F_{\boldsymbol{u}_{i}} \cdot \nabla G_{n_i}) \\ + \frac{\nabla \times \boldsymbol{u}_{\sigma}}{m_\sigma n_\sigma }\cdot (F_{\boldsymbol{u}_{\sigma}}\times G_{\boldsymbol{u}_{\sigma}})\,d^3\boldsymbol{x} \\ \int \frac{1}{\epsilon \epsilon_o}(F_{\boldsymbol{E}} \cdot (\nabla\times{G_{\boldsymbol{B}}}) - G_{\boldsymbol{E}} \cdot (\nabla\times{F_{\boldsymbol{B}}}))\,d^3\boldsymbol{x}  \\+ \int  \frac{q_\sigma}{m_\sigma \epsilon^2 \epsilon_o}(F_{\boldsymbol{u}_{\sigma}}\cdot G_{\boldsymbol{E}} - F_{\boldsymbol{E}}\cdot G_{\boldsymbol{u}_{\sigma}})\, d^3\boldsymbol{x} \\ + \int \frac{q_\sigma \boldsymbol{B}}{\epsilon(m_\sigma)^2 n_\sigma}  \cdot (F_{\boldsymbol{u}_{\sigma}}\times G_{\boldsymbol{u}_{\sigma}})\,d^3\boldsymbol{x}
    \end{split}
\end{align}
where $n_e(n_i, \boldsymbol{E}) = Z_in_i + \epsilon^2 \epsilon_oq_e^{-1}\nabla \cdot \boldsymbol{E}$ is understood as a function of $n_i$ and $\boldsymbol{E}$. A quick computation reveals that the above bracket recovers equations \eqref{ion_momentum} - \eqref{faraday} using the Hamiltonian given in Ref.~\onlinecite{10.1063/1.4994068}.

    \section{Two Fluid Maxwell to MHD}

    We now give our final example of Poisson-Dirac constraint theory. We write the two-fluid-Maxwell phase space as the product $M =X\times Y$, where $X\ni (n_i,\bm{u}_i,\bm{B})$ and $Y\ni (\bm{u}_e,\bm{E})$. The constraint set $N\subset X\times Y$ will be given by 
    \begin{equation}
    N = \{(n_i,\boldsymbol{u}_i,\boldsymbol{B},\boldsymbol{u}_e,\boldsymbol{E}):\boldsymbol{u}_e=\boldsymbol{u}_i, \quad \boldsymbol{E} = -\boldsymbol{u}_i\times \boldsymbol{B}\},
    \end{equation}
    as in the slow manifold analysis from Ref.\,\onlinecite{Burby_2020}.
    We proceed by demonstrating that $N$ is a transversal with respect to $(\pi, M)$. 
    \begin{thm}
        $N$ is a Poisson transversal with respect to $(\pi, M)$. 
    \end{thm}
    \begin{proof}
        We proceed by computing $T_pN$ and proving the transversality condition. $T_pN$ is given by 
        \begin{equation}
        \begin{split}
            T_pN = \{(\delta x, \delta y )\in M: \delta \boldsymbol{u}_e = \delta \boldsymbol{u}_i, \\\delta \boldsymbol{E} = - \delta \boldsymbol{u}_i\times \boldsymbol{B} - \boldsymbol{u}_i\times \delta \boldsymbol{B}\}.
        \end{split}
        \end{equation}
        $T_pN$ spans the tuple $(n_i, B, u_i)$ so we only need check that $T_pL_p$ spans the fast variables. To this end, let $(\boldsymbol{u_e}, \boldsymbol{E})\in Y$ be arbitrary. 

        We may pick the one form $\theta = \int \frac{\epsilon^2 \epsilon_o m_i}{Z_iq_e}(\boldsymbol{E}\cdot \delta \boldsymbol{u_i}) + \frac{\epsilon^2 \epsilon_o m_e}{q_e}(\boldsymbol{u_e}\cdot \delta \boldsymbol{E})\,d^3\boldsymbol{x}$ so that $\pi^\#(\theta) = (\boldsymbol{u}_e, \boldsymbol{E})$ so that the transversality condition is satisfied. It is very difficult to prove that $N$ is a Poisson transversal for arbitrary $\epsilon$. However, we can prove that the condition holds {  formally in the the limit $\epsilon \to 0$}. We proceed by expanding $\pi$ perturbatively
        \begin{equation}
            \epsilon^2\pi_\epsilon = \sum_{i=0}^2 \epsilon^i\pi_i
        \end{equation}
        where
        \begin{equation}
            \pi_0(\delta F) = \begin{pmatrix}
            0 \\
            - \frac{q_e}{\epsilon_o m_e}F_{\boldsymbol{u}_e} +  \frac{Z_iq_e}{\epsilon_o m_i}F_{\boldsymbol{u}_i} \\
            -\frac{Z_iq_e}{ \epsilon_o m_i}F_{\boldsymbol{E}}  \\ 
            \frac{q_e}{ \epsilon_o m_e}F_{\boldsymbol{E}}  \\
            0\\
        \end{pmatrix},
        \end{equation}
        \begin{equation}
            \pi_1(\delta F) = \begin{pmatrix}
            -\frac{1}{ \epsilon_o}\nabla\times{F_{\boldsymbol{E}}} \\
            \frac{1}{ \epsilon_o}\nabla\times{F_{\boldsymbol{B}}}  \\
            \frac{-Z_iq_e}{m_i^2n_i}(F_{\boldsymbol{u}_i}\times \boldsymbol{B}) \\ 
            \frac{q_e}{m_e^2n_e}(F_{\boldsymbol{u}_e}\times \boldsymbol{B}) \\
            0
            \end{pmatrix}
        \end{equation}
        and,
        \begin{equation}
            \pi_2(\delta F) = \begin{pmatrix}
            0 \\
            0\\
            -m_i^{-1}\nabla F_{n_i} + \frac{1}{m_i n_i} F_{\boldsymbol{u}_{i}} \times (\nabla \times \boldsymbol{u}_{i}) \\ 
             \frac{1}{m_e n_e} F_{\boldsymbol{u}_{e}} \times (\nabla \times \boldsymbol{u}_{e}) \\
            - m_i^{-1}\nabla\cdot{F_{\boldsymbol{u}_i}} \\
        \end{pmatrix}.
        \end{equation}
        Moving on to the intersection condition, suppose towards a contradiction that there is a vector $\nu\in T_pN\cap (T_pN)^{\perp\pi}$. We proceed by characterising the annihilator space $(T_pN)^{\perp\pi}$. This is achieved by applying an arbitrary element of the annihilator space to $\pi^\#$: 
        \begin{equation}
        \begin{split}
            \delta F = \langle \boldsymbol{B}\times F_{\boldsymbol{E}} -F_{\boldsymbol{u}_e}, \delta \boldsymbol{u}_i\rangle + \langle -\boldsymbol{u}_i\times F_{\boldsymbol{E}}, \delta \boldsymbol{B} \rangle \\+ \langle F_{\boldsymbol{E}}, \delta \boldsymbol{E}\rangle +  \langle F_{\boldsymbol{u}_e}, \delta \boldsymbol{u}_e\rangle 
        \end{split}
        \end{equation}
    Checking the intersection condition perturbatively, we write $\delta F_\epsilon \in (T_pN)^\circ$ where $\delta F_\epsilon = \sum_i \epsilon^i\delta F_i$, apply $\pi_\epsilon$, and apply the condition $T_pN\cap (T_pN)^{\perp \pi}$. The equations order by order are
    \begin{equation}
    \begin{split}
        \pi_\epsilon(\delta F_\epsilon) = \pi_0(\delta F_0) + \epsilon (\pi_1(\delta F_0) + \pi_0(\delta F_1)) + \epsilon^2(\pi_2(\delta F_0) \\+ \pi_0(\delta F_2) + \pi_1(\delta F_1)) + \order{\epsilon^3}
    \end{split}
    \end{equation}
    and
    \begin{equation} 
            \pi_0(\delta F_0) = \begin{pmatrix}
            0 \\
            - \frac{q_e}{\epsilon_o m_e}F_{\boldsymbol{u}_e,0} +  \frac{Z_iq_e}{\epsilon_o m_i}(\boldsymbol{B}\times F_{\boldsymbol{E},0} -F_{\boldsymbol{u}_e ,0}) \\
            -\frac{Z_iq_e}{ \epsilon_o m_i}F_{\boldsymbol{E},0}  \\ 
            \frac{q_e}{ \epsilon_o m_e}F_{\boldsymbol{E},0}  \\
            0\\
    \end{pmatrix}.
    \end{equation}
    Applying the condition $\delta \boldsymbol{u}_i = \delta \boldsymbol{u}_e$ implies that $F_{\boldsymbol{E},0}=0$. Substituting this back into above equation and applying the condition $\delta \boldsymbol{E} = -\delta \boldsymbol{u}_i \times \boldsymbol{B} - \boldsymbol{u}_i \times \delta \boldsymbol{B}$, we find that $F_{\boldsymbol{u}_e ,0}=0$. Thus we have that all $F_0$ terms in equation X vanish so that 
    \begin{equation}
        \pi_\epsilon(\delta F_\epsilon) = \epsilon ( \pi_0(\delta F_1)) + \epsilon^2(\pi_0(\delta F_2) + \pi_1(\delta F_1)) + \order{\epsilon^3}.
    \end{equation}
    Moving up to order $\epsilon$, we notice that by a similar argument to before, $\delta F_1$ also vanishes. Now that the base cases are satisfied, we demonstrate inductively that $\delta F_n$ vanishes for any $n\in \N$. Suppose that $\delta F_n$ vanishes for $n'\leq n$. Then, at order $\epsilon^{n+1}$, the equation for an element of $(T_pN)^{\perp \pi}$ is written 
    \begin{equation}
        \pi_\epsilon(\delta F_\epsilon)=\pi_{0}(\delta F_{n+1}).
    \end{equation}
    Applying the condition that $\pi_{0}(\delta F_{n+1})\in T_pN$ implies that $\delta F_{n+1}=0$ by the base case. Thus at any order in $\epsilon$, the condition $T_pN\cap(T_pN)^{\perp \pi}=\{0\}$ is satisfied.
    \end{proof}
    \subsection{Derivation of the Reduced Two Fluid bracket}

    In order to compute the slow bracket, we find extensions $\delta \tilde{F}\in T^\ast M$ of functionals $\delta F\in T^\ast N$ using the ansatz (\eqref{Extension_ansatz})
    \begin{align}
    \begin{split}
        \delta \tilde{F}=\int F_{n_i}\cdot\delta n_i + (F_{\boldsymbol{u}_{i}} - \Theta_{\boldsymbol{E}}\times \boldsymbol{B} - \Theta_{\boldsymbol{u}_e} )\cdot \delta \boldsymbol{u}_{i}\\ + (\Theta_{\boldsymbol{E}}\times \boldsymbol{u}_{i}+ F_{\boldsymbol{B}})\cdot \delta \boldsymbol{B} \\+ \int \Theta_{\boldsymbol{u}_e}\cdot \delta \boldsymbol{u}_e + \Theta_{\boldsymbol{E}}\cdot \delta \boldsymbol{E} \,d^3\boldsymbol{x}.
    \end{split}
    \end{align}
    In order to solve perturbatively, we expand the unknown functions in powers of $\epsilon$:
    \begin{equation}
        \begin{split}
            \Theta_{\boldsymbol{E}} = \Theta^0_{\boldsymbol{E}} + \epsilon \Theta^1_{\boldsymbol{E}} +  \epsilon^2\Theta^2_{\boldsymbol{E}} +\dots \\
            \Theta_{\boldsymbol{u}_e} = \Theta^0_{\boldsymbol{u}_e} + \epsilon \Theta^1_{\boldsymbol{u}_e} +  \epsilon^2\Theta^2_{\boldsymbol{u}_e } + \dots
        \end{split}
    \end{equation}
    where $\Theta\in Hom(X,Y)$ and $\delta x\in X$. Although the extensions will formally be an infinite series in $\epsilon$, we need only terms up to order $\epsilon^2$ since we seek only the $\order{1}$ bracket on $N$ (recall that the bracket was scaled by $\epsilon^2$). We now look to apply $\pi^\#(\delta \tilde{F})$ and enforce the condition $\pi^\#(\delta \tilde{F})\in T_pN$:
    \begin{widetext}
        \begin{equation}
    \begin{split}
        &\pi_{N}(\delta \tilde{F})=\\ &\begin{pmatrix}
            -\frac{1}{\epsilon \epsilon_o}\nabla\times{\Theta_{\boldsymbol{E}}} \\
            \frac{1}{\epsilon \epsilon_o}\nabla\times{(\Theta_{\boldsymbol{E}}\times \boldsymbol{u}_{i}+ F_{\boldsymbol{B}})} - \frac{a_e}{\epsilon \epsilon_o m_e}\Theta_{\boldsymbol{u}_e} -  \frac{a_i}{\epsilon \epsilon_o m_i}(F_{\boldsymbol{u}_{i}} - \Theta_{\boldsymbol{E}}\times \boldsymbol{B} - \Theta_{\boldsymbol{u}_e} ) \\
            \frac{a_i}{m_i^2n_i}((F_{\boldsymbol{u}_{i}} - \Theta_{\boldsymbol{u}_e} )\times \boldsymbol{B}) +\frac{a_i}{\epsilon \epsilon_o m_i}\Theta_{\boldsymbol{E}}-m_i^{-1}\nabla F_{n_i} - \frac{1}{m_i n_i} (F_{\boldsymbol{u}_{i}} - \Theta_{\boldsymbol{E}}\times \boldsymbol{B} - \Theta_{\boldsymbol{u}_e} )\times (\nabla \times \boldsymbol{u}_{i}) \\ 
            \frac{a_e}{m_e^2n_e}(\Theta_{\boldsymbol{u}_e}\times \boldsymbol{B}) +\frac{a_e}{\epsilon \epsilon_o m_e}\Theta_{\boldsymbol{E}} - \frac{1}{m_e n_e} \Theta_{\boldsymbol{u}_e} \times (\nabla \times \boldsymbol{u}_{e}) \\
            - m_i^{-1}\nabla \cdot {(F_{\boldsymbol{u}_{i}} - \Theta_{\boldsymbol{E}}\times \boldsymbol{B} - \Theta_{\boldsymbol{u}_e} )} \\
        \end{pmatrix}
    \end{split}
    \end{equation}
    \end{widetext}
    where $n_e = n_e(n_i, \boldsymbol{E})$ is a function of $n_i$ and $\boldsymbol{E}$. Solving perturbatively, we find that the solution at order zero is
    \begin{align}
        &\Theta^0_{\boldsymbol{E}} = 0\\
        &\Theta^0_{\boldsymbol{u}_e} = \frac{Z_i\nu}{(1+ Z_i\nu)}F_{\boldsymbol{u}_i}.
    \end{align}
    
    At order 1, we have 
    \begin{align}
        \Theta^1_{\boldsymbol{E}} &=  -\frac{Z_i \epsilon_0}{(1+\nu Z_i)^2}(\frac{1}{m_in_i}+ \frac{1}{m_en_e})F_{\boldsymbol{u}_{i}}\times \boldsymbol{B} \\
         \begin{split} \Theta^1_{\boldsymbol{u}_e} & =
             \frac{\nu m_i(1+\nu Z_i)}{q_e}\nabla\times F_{\boldsymbol{B}} - \frac{Z_i^2 \epsilon_0}{(1+\nu Z_i)}(\frac{1-\nu}{m_in_i}\\& - \frac{2\nu}{m_en_e})(F_{\boldsymbol{u}_{i}}\times \boldsymbol{B})\times \boldsymbol{B}.
        \end{split}
    \end{align}
    One subtle issue here is that $n_e = Z_in_i + \epsilon^2 \epsilon_oq_e\nabla \cdot \boldsymbol{E}$ is epsilon dependent, so the equations above have extra, hidden $\epsilon$ dependence. In order to resolve this, we expand $1/n_e$ as a power series in $\epsilon$:
    \begin{equation}
        \frac{1}{n_e} = \sum_{n=0}^\infty\frac{(-1)^{n}(\epsilon^2 \epsilon_0 q_e^{-1}\nabla\cdot \boldsymbol{E})^{n}}{(Z_in_i)^{n}} = \frac{1}{Z_in_i} + \order{\epsilon^2}. 
    \end{equation}

    The first term in the expansion is $\epsilon^2$, so that terms with a factor of $1/n_e$ are two orders higher in epsilon. As a result, we remove any terms a factor of $1/n_e$ as they are at least order $\epsilon^3$ in the above expression. Correcting, we find that 
    \begin{align}
        \Theta^1_{\boldsymbol{E}} &= -\frac{\epsilon_0}{\mathsf{m}n_i}F_{\boldsymbol{u}_{i}}\times \boldsymbol{B} \\
        \Theta^1_{\boldsymbol{u}_e} & = \frac{m_e}{q_e(1 + \nu Z_i)}\nabla\times F_{\boldsymbol{B}} - \frac{Z_i m_e\epsilon_0}{\mathsf{m}^2n_i}(F_{\boldsymbol{u}_{i}}\times \boldsymbol{B})\times \boldsymbol{B}
    \end{align}
    where $\mathsf{m} = (1+\nu Z_i)m_i$. In order to obtain the bracket at order $1$ in epsilon, one would expect that the second order extension function would be needed. However, for reasons that will become more apparent later, the term involving the second order extension function in the Poisson bracket will vanish.
    Thus, the extensions up to order $\epsilon$ are written
    \begin{equation}
        \delta\tilde{F}_\epsilon = \delta\tilde{F}_0+\epsilon\delta\tilde{F}_1
    \end{equation}
    where
    \begin{equation}
    \begin{split}
       \delta  \tilde{F}_1=\int F_{n_i}\delta n_i + (F_{\boldsymbol{u}_{i}}  - \Theta^0_{\boldsymbol{u}_e} )\cdot \delta \boldsymbol{u}_{i} \\+  F_{\boldsymbol{B}}\cdot \delta \boldsymbol{B}\, d^3\boldsymbol{x} + \int \Theta^0_{\boldsymbol{u}_e} \cdot \delta \boldsymbol{u}_e \, d^3\boldsymbol{x}
    \end{split}
    \end{equation}
    and
    \begin{equation}
    \begin{split}
        \delta \tilde{F}_1=\int -(  \Theta^1_{\boldsymbol{E}}\times \boldsymbol{B}  +  \Theta^1_{\boldsymbol{u}_e} )\cdot \delta \boldsymbol{u}_{i}\,d^3\boldsymbol{x} \\+ (\Theta^1_{\boldsymbol{E}}\times \boldsymbol{u}_{i} )\cdot \delta \boldsymbol{B} \,d^3\boldsymbol{x}+ \int  \Theta^1_{\boldsymbol{u}_e} \cdot \delta \boldsymbol{u}_e + \Theta^1_{\boldsymbol{E}}  \cdot \delta \boldsymbol{E} \, d^3\boldsymbol{x}.
    \end{split}
    \end{equation}
    Now we explicitly write the Poisson structure on $N$. Applying these extensions to the bracket, we write $\pi^\#(\delta \tilde{F})$ perturbatively:
    \begin{equation}
        \epsilon^2 \pi(\delta \tilde{F}) = \sum_{j=0}^2\epsilon^j\pi^j(\delta \tilde{F}) . 
    \end{equation}
    In this paper, we aim to show that the order $1$ bracket reproduces MHD so that higher order terms in the bracket are not needed.
    A trivial computation gives $\langle \delta \tilde{G}_0, \pi^\#(\delta \tilde{F}_0)\rangle=0$. Notice that since the Poisson tensor is skew symmetric, and $\pi_0(\delta \tilde{F}_0)=0$, we find that
     \begin{equation}
         \langle \delta \tilde{G}_0, \pi_0(\delta \tilde{F}_0) \rangle, \langle \delta \tilde{G}_0, \pi_0(\delta \tilde{F}_1) \rangle,\langle \delta \tilde{G}_0, \pi_0(\delta \tilde{F}_2) \rangle  = 0.
     \end{equation}
     We contract this vector field with an extended form $\delta \tilde{G}$ of $\delta G$ defined similarly to $\delta F$ and collect the terms order by order. By the previous argument, all order zero terms are zero. The only non zero order $\epsilon$ term is $\langle \delta\tilde{G}_0, \pi_1(\delta \tilde{F}_0) \rangle$. Explicitly, this term is 
    \begin{equation}
    \begin{split}
        \langle \delta\tilde{G}, \pi(\delta \tilde{F})\rangle = \\\int   \frac{-Z_iq_i}{m_i^2n_i}(G_{\boldsymbol{u}_i} - \Theta^0_{G,\boldsymbol{u}_e})\cdot ([F_{\boldsymbol{u}_i}-\Theta^0_{F, \boldsymbol{u}_e}]\times B)\\ + \frac{q_e}{m_e^2n_e}\Theta^0_{G,\boldsymbol{u}_e}\cdot \Theta^0_{F,\boldsymbol{u}_e}\,d^3\boldsymbol{x}.
    \end{split}
    \end{equation}
    Using the simplifications due to $1/n_e = 1/(Z_i n_i)$ at $\order{1}$, we find that the first order terms vanish:
    \begin{equation}
    \begin{split}
        \langle \delta\tilde{G}, \pi(\delta \tilde{F})\rangle= \int   \frac{-Z_iq_e}{\mathsf{m}^2n_i}G_{\boldsymbol{u}_i}\cdot (F_{\boldsymbol{u}_i}\times B) \\+ \frac{Z_iq_e}{\mathsf{m}^2n_i}G_{\boldsymbol{u}_i}\cdot (F_{\boldsymbol{u}_i}\times \boldsymbol{B}) \,d^3\boldsymbol{x} = 0.
    \end{split}
    \end{equation}
    The order 2 terms are $\langle \delta\tilde{G}_1, \pi_0(\delta \tilde{F}_1)\rangle + \langle \delta\tilde{G}_1, \pi_1(\delta \tilde{F}_0)\rangle + \langle \delta\tilde{G}_0, \pi_2(\delta \tilde{F}_0)\rangle + \langle \delta\tilde{G}_0, \pi_1(\delta \tilde{F}_1)\rangle$. Starting with $\langle \delta\tilde{G}_0, \pi_2(\delta \tilde{F}_0)\rangle$, we find that 
    \begin{equation}
    \begin{split}
       & \langle \delta \tilde{G}_0, \pi_2(\delta \tilde{F}_0)\rangle =\\ &\int -\mathsf{m}^{-1}(G_{\boldsymbol{u}_i}\cdot \nabla F_{n_i} + G_{n_i}\nabla \cdot F _{\boldsymbol{u}_i}) \\&+ \frac{1}{(1+\nu Z_i)^2}( \frac{(\nabla\times \boldsymbol{u}_i)}{m_in_i}\\&+\frac{(\nabla\times \boldsymbol{u}_e)(Z_i \nu)^2}{m_en_e} ) \cdot (G_{\boldsymbol{u}_i}\times F_{\boldsymbol{u}_i} )\, d^3\boldsymbol{x}.
    \end{split}
    \end{equation}
    Upon restriction to $N$, we find 
    \begin{equation}
    \begin{split}
        &\langle \delta \tilde{G}_0, \pi_2(\delta \tilde{F}_0)\rangle\\& = \int -\mathsf{m}^{-1}(G_{\boldsymbol{u}_i}\cdot \nabla F_{n_i} + G_{n_i}\nabla \cdot F _{\boldsymbol{u}_i}) \\&+ \frac{(\nabla\times \boldsymbol{u}_i)}{\mathsf{m}n_i} \cdot (G_{\boldsymbol{u}_i}\times F_{\boldsymbol{u}_i} ) \,d^3\boldsymbol{x}.
    \end{split}
    \end{equation}
    Moving onward, we consider the $\langle \delta \tilde{G}_1, \pi_1(\tilde{F}_0)\rangle + \langle \delta \tilde{G}_0, \pi_1(\tilde{F}_1)\rangle$ term. Considering the sum of these terms is advantageous because of the skew symmetry: 
    \begin{equation}
    \begin{split}
        \langle \delta\tilde{G}_1, \pi_1(\tilde{F}_0)\rangle + \langle \tilde{G}_0, \pi_1(\tilde{F}_1)\rangle =\\ \int \frac{1}{\mathsf{m}n_i}\boldsymbol{B}\cdot(G_{\boldsymbol{u}_i}\times (\nabla\times F_{\boldsymbol{B}}) - F_{\boldsymbol{u}_i}\times (\nabla\times G_{\boldsymbol{B}})) \\- \frac{q_e}{\epsilon_0m_e}(\tilde{F}^1_{\boldsymbol{E}}\cdot\tilde{G}_{\boldsymbol{u}_e}^1-\tilde{G}^1_{\boldsymbol{E}}\cdot\tilde{F}_{\boldsymbol{u}_e}^1 )\\ - \frac{Z_iq_e}{\epsilon_0m_i}(\tilde{G}^1_{\boldsymbol{E}}\cdot\tilde{F}_{\boldsymbol{u}_i}^1-\tilde{F}^1_{\boldsymbol{E}}\cdot\tilde{G}_{\boldsymbol{u}_i}^1)\, d^3\boldsymbol{x}.
    \end{split}
    \end{equation}
    Finally, we compute $\langle \delta\tilde{G}_1, \pi_0(\delta\tilde{F}_1)\rangle$:
    \begin{equation}
    \begin{split}
        \langle \tilde{G}_1, \pi_0(\tilde{F}_1)\rangle =\int \frac{q_e}{\epsilon_0m_e}(\tilde{F}^1_{\boldsymbol{E}}\cdot\tilde{G}_{\boldsymbol{u}_e}^1-\tilde{G}^1_{\boldsymbol{E}}\cdot\tilde{F}_{\boldsymbol{u}_e}^1 ) \\+ \frac{Z_iq_e}{\epsilon_0m_i}(\tilde{G}^1_{\boldsymbol{E}}\cdot\tilde{F}_{\boldsymbol{u}_i}^1-\tilde{F}^1_{\boldsymbol{E}}\cdot\tilde{G}_{\boldsymbol{u}_i}^1)\,d^3\boldsymbol{x}.
    \end{split}
    \end{equation}
    Putting this all together, we find the Poisson bracket on $N$ to be 
    \begin{equation}
    \begin{split}
        \pi(\delta G,\delta F) = \int -\mathsf{m}^{-1}(G_{\boldsymbol{u}_i}\cdot \nabla F_{n_i} + G_{n_i}\nabla \cdot F_{\boldsymbol{u}_i}) \\+ \frac{(\nabla\times \boldsymbol{u}_i)}{\mathsf{m}n_i} \cdot (G_{\boldsymbol{u}_i}\times F_{\boldsymbol{u}_i}) \\+\frac{1}{\mathsf{m}n_i}\boldsymbol{B}\cdot(G_{\boldsymbol{u}_i}\times (\nabla\times F_{\boldsymbol{B}}) - F_{\boldsymbol{u}_i}\times (\nabla\times G_{\boldsymbol{B}})) \,d^3\boldsymbol{x}.
    \end{split}
    \end{equation}
    This bracket coincides with the bracket given by Ref.~\onlinecite{10.1063/1.4994068} in the limit $\epsilon=0$.

\section{Discussion}

The Poisson-Dirac constraint method provides a more general, coordinate free approach to applying dissipation free constraints to plasma models. The wider applicability of the PD constraint model allows for the construction of novel reduced plasma models. In particular, the PD constraint method could be applied to the multi-fluid problem in Ref.~\onlinecite{Dewar_2020} to derive a Poisson bracket on the constrained space while avoiding aforementioned issues with inconsistent equations of motion. In a similar fashion to \cite{10.1063/1.4994068}, this new constraint method can also be used to explain why known reduced plasma models have Hamiltonian structure.

It is well known that the two fluid equations arise as a reduction by symmetry \cite{10.1063/1.4994068} so that the Poisson bracket on the ITFM space admits a symplectic realisation (see Ref.~\onlinecite{Vaisman1994_symp, 32155} for details on symplectic realisation). By performing the reduction of bracket for the two fluid slow manifold using the symplectic realisation, Ref.\,\onlinecite{10.1063/1.4994068} was able to find a closed form expression for the Poisson bracket. This presents a tension between the Poisson and symplectic methods for reduction since the bracket obtained via the PD constraint method yields a bracket that is formally an infinite series in $\epsilon$ (the bracket contains terns with a factor of $1/n_e$ making it formally an infinite series in $\epsilon$). While the leading-order terms in this expansion are guaranteed to define a genuine Poisson bracket, truncating the series at higher orders may violate the Jacobi identity. It may be possible to deform the Poisson bracket into the bracket in Ref.~\onlinecite{10.1063/1.4994068} via Poisson automorphisms. The obstruction to finding such an automorphism is represented as a co-homology class in the Poisson co-homology (see Ref.~\onlinecite{32155, Vaisman1994} for details). Unfortunately, the Poisson co-homology is notoriously difficult to compute. In fact, the co-homology groups are often infinite dimensional. In the literature, the co-homology has been computed under assumptions that are too restrictive for application for plasma physics problems\cite{Xu1992}. Fernandes et al\cite{fernandes2009momentummappoissongeometry} note that \emph{any} Poisson manifold is integrable by Weinstein symplectic Groupoid. It may be possible to employ an approach similar to Ref.~\onlinecite{Xu1992} using the more general structure in Ref.~\onlinecite{fernandes2009momentummappoissongeometry, tseng2004integratingpoissonmanifoldsstacks} to establish the equivalence of the Poisson and symplectic reduction methods. This issue will be resolved in a future paper.

\section{Acknowledgements}
This work was supported by U.S. Department of Energy grant \# DE-FG02-04ER54742.

{ 
\appendix*
\section{Lifting Poisson bivectors off of Submanifolds}
As seen in example 2, it is not always possible to extend a Poisson bivector on a submanifold to an open neighborhood around the submanifold. Here, we show the conditions under which it is possible to lift a Poisson bivector off of a submanifold. 
\begin{thm}\label{Local_trans_bivector}
    Suppose that a Poisson transversal $(N, \pi_N)$ of a Poisson manifold $(M, \pi)$. Then, in an open neighborhood $U$ around $N$, there is a foliation of $U$ by Poisson transversals. 
\end{thm}
\begin{proof}
    Using the normal form for Poisson transversals\cite{32155}, there exists a local splitting of $\pi$ into the product $(N,\pi_N)\times(S,\pi_S)$ where  $(S,\pi_S)$ is a symplectic manifold and $(N,\pi_N)$ is a Poisson transversal. The fibres of the projection onto the first factor constitute a foliation by Poisson transversals. Explicitly, the foliation is given by
    \begin{equation}
        N_p = N\times \{p\}
    \end{equation}
    for $p\in S$
\end{proof}
The normal form theorem does not explicitly use local coordinates or refer to dimension, so there are no formal obstructions to an infinite dimensional version of the above theorem. That being said, the existence of tubular neighborhoods in this context fail to exist unless some strong regularity (e.g Hilbert Manifolds) is imposed. Additionally, the proof relies on finding flows for time dependent vector fields, the existence of which requires additional structure. We extend this proof the case of Poisson-Dirac submanifolds under the assumption that the $\pi$-orthogonal has constant rank. Such manifolds are called coregular Poisson-Dirac submanifolds\cite{32155}. 
\begin{definition}
    A submanifold $N$, of a Poisson manifold $(M,\pi)$ is called a \textbf{coregular Poisson-Dirac} submanifold if 
    \begin{enumerate} [label=\textbf{C.\arabic*}]
        \item $T_pN\cap T_pN^{\perp\pi} = \{0\}$ 
        \item $T_pN^{\perp\pi}$ has constant rank
    \end{enumerate}
\end{definition}
Notice that condition \textbf{C.1} implies that the induced bivector on $N$ is smooth (see Ref.~\onlinecite{32155}, Prop 8.42, page 169). The constant rank assumption implies that coregular Poisson-Dirac submanifolds are Poisson submanifolds inside of Poisson Transversals\cite{32155}
\begin{lemma} \label{coreg_PD_Poi_trans}
    Given a Poisson manifold $(M,\pi)$ and a Poisson-Dirac submanifold $N$, the following are equivalent:
    \begin{enumerate}
        \item $N$ is a coregular Poisson-Dirac submanifold
        \item  $N$ is a Poisson Submanifold of some Poisson transversal, X. 
    \end{enumerate}
    Moreover, the germ of X around N is unique up to Poisson diffeomorphism.
\end{lemma}
\begin{proof}
    See Ref.~\onlinecite{32155} page 170 for proof.
\end{proof}
Moreover, the symplectic leaves of coregular Poisson-Dirac submanifolds are the connected components of the intersection with the symplectic leaves of the ambient space. As a corollary from the above lemma, and \ref{Local_trans_bivector}, we show that the Poisson bivector on a coregular Poisson Dirac submanifold $N$ can be extended to an open neighborhood around $N$. 
\begin{corollary}
    Let $N$ be a coregular Poisson Dirac submanifold  of a Poisson manifold $(M, \pi)$. For any open neighborhood $U\subset M$ of $N$, there exists a foliation of $U$ by coregular Poisson Dirac submanifolds, $\{N_p\}$ with $N_0 = N$. 
\end{corollary}
\begin{proof}
    By Lemma \ref{coreg_PD_Poi_trans}, there exists a Poisson transversal $(X,\pi_X)$ to which $(N,\pi_N)\subset (X,\pi_X)$ is a Poisson submanifold of $X$. By theorem \ref{Local_trans_bivector}, in an open neighborhood around $X$, there exists a splitting, $(X, \pi_X)\times (S,\pi_S^{-1})$. Explicitly, the foliation from the splitting is given by $\{X_p\} = X\times \{p\}$ for varying $p\in S$. $N$ pulls back the Poisson submanifold $N\times \{p\}$. Thus, we find that $N\subset X$ lifts to a foliation 
    \begin{equation}
        N_p = X_p\cap U,
    \end{equation}
    where $U\subset M$ is an open neighborhood around $N$ in the union of the local transversal foliation $X_p$. For each $p\in S$, $N_p\subset X_p$ is a Poisson submanifold. Thus, we have the local foliation around $N$ by coregular Poisson-Dirac submanifolds as required. 
\end{proof}
\begin{remark}
    Given that the germ of of $X$ is unique up to Poisson diffeomorphism, the local foliation by coregular Poisson-Dirac submanifolds is unique up to Poisson diffeomorphism.
\end{remark}
Using the above lemma, we are able to extend $\pi_N$ to an open neighborhood by constructing a Poisson transversal in the local foliation and defining a family of bivectors on each leaf, $\{\pi_{N_p}\}$. In the case of Poisson Dirac submanifolds, the non constant rank implies that it is not possible to define a foliation around any neighborhood of $N$ where the rank of $TN^{\perp \pi}$ changes. Indeed, the foliation around the rank change is singular, as seen in example 2. We prove below that lifting Poisson bivectors off of submanifolds fails exactly when the rank of $TN^{\perp \pi}$ is locally non constant.
\begin{thm}
    Let $N\subset M$ be a Poisson-Dirac submanifold, and $U\subset M$ be any neighborhood where the rank of $T_pN^{\perp\pi}$ is non constant. Then there is no smooth foliation of $U$ by Poisson-Dirac Submanifolds $\{N_\alpha\}_{\alpha \in \lambda}$ such that each $N_\alpha$ and $N_\beta$ are Poisson Diffeomorphic.
\end{thm}
\begin{proof}
    Suppose towards a contradiction that $\tilde{\pi}$ is a smooth Poisson bivector defined on an open $U\subset M$ neighborhood where the rank of $T_pN^{\perp\pi}$ is non constant, and that $\tilde{\pi}\mid_N = \pi_N$. Given the assumption of a rank dropping point in $N$, there exists a point $p_0$ and a sequence $\{p_n\}$ such that
    \begin{equation}
        \lim_{n\to \infty}\inf \text{rk} (\pi_N(p_n)) <rank(\pi_N(p_0))
    \end{equation}
    Given each leaf of the foliation are Poisson diffeomorphic, there exists a rank dropping point $p_\alpha\in N_\alpha$ for each $N_\alpha$. By the linear theory of Poisson-Dirac submanifolds\cite{INT_PB},
    \begin{equation}
        \text{codim}\, N - \text{rk} \, T_pN^\circ\cap\ker \pi_p^\#= \text{rk}\, \pi_{p}- \text{rk}\, \pi_{N,p}
    \end{equation}
    Given that 
    \begin{equation}
        \text{codim}\, N - \text{rk} \, T_pN^\circ\cap\ker = \text{rk} \, T_pN^{\pi \perp}
    \end{equation}
    We find that 
    \begin{equation}
        \text{rk} \, T_pN^{\pi \perp} + \text{rk}\, \pi_{N,p}= \text{rk}\, \pi_{p}
    \end{equation}
    Thus we consider a sequence of points $\{p_\alpha: p_\alpha\in N_\alpha\}$ such that $\text{rk}\, \pi_{N_\alpha,p_\alpha} = r \in \R$ and $\text{rk}\, T_{p_\alpha}N_\alpha^{\pi \perp} = k \in \R$ are constant, and converges to a singular point of $N$ $\text{rk}\, T_{\tilde{p}}N^{\pi \perp} = k' < k \in \R$. Then it follows that for all $p_\alpha$,
    \begin{equation}
        \text{rk}\, \tilde{\pi}_{p_\alpha} = r + k
    \end{equation}
    But, at $\tilde{p}$,
    \begin{equation}
        \text{rk}\, \tilde{\pi}_{\tilde{p}} = r + k'
    \end{equation}
    thus,
    \begin{equation}
        \lim_{\alpha \to \infty}\text{rk}\, \tilde{\pi}_{p_\alpha} = r + k > r + k'
    \end{equation}
    Which implies that the rank of $\tilde{\pi}$ jumps at $\tilde{p}$, which contradicts the smoothness of $\tilde{\pi}$.
\end{proof}
The above theorem implies that it is not possible to construct a family of bivectors, $\{\pi_{N_p}\}$, so that no local extension, $\tilde{\pi}$ is possible.
\begin{definition}
    A function $f:M\to \N$ on a topological space $M$ is said to be \textbf{lower semi continuous} if for all $k\in \N$, the set $\{p\in M: f(p)\geq k\}$ is open in $M$
\end{definition}
\begin{definition}
    Given a matrix $A(p): \R^{n}\to \R^m$ smoothly depending on $p\in M$ for $M$ a $n$ dimensional manifold, the rank of $A$ is said to \textbf{drop} at $p_0\in M$ if there exists a sequence $p_n$ such that 
    \begin{equation}
        \lim_{n\to \infty}\sup \text{rk} (A(p_n)) > rank(A(p_0))
    \end{equation}
\end{definition}
}
\providecommand{\noopsort}[1]{}\providecommand{\singleletter}[1]{#1}%
%


\end{document}